\documentclass[a4paper,10pt]{article}
\usepackage{amsfonts}
\usepackage{amscd,color}
\usepackage{amsmath,amsfonts,amssymb,amscd}
\usepackage{indentfirst,graphicx,epsfig}
\usepackage{graphicx}
\input{epsf}
\usepackage{graphicx,indentfirst}
\usepackage{epstopdf}
\usepackage{caption}
\usepackage[ruled]{algorithm2e}
\usepackage{enumerate}
\usepackage{xcolor}
\usepackage{calligra}
\usepackage{etoolbox}
\usepackage{tikz}
\usepackage{algorithmicx}
\usepackage{algpseudocode}
\usepackage[colorlinks,linkcolor=blue,anchorcolor=blue,citecolor=blue]{hyperref}

\newtheorem{theorem}{Theorem}[section]
\newtheorem{definition}[theorem]{Definition}
\newtheorem{problem}[theorem]{Problem}

\newtheorem{lemma}[theorem]{Lemma}
\newtheorem{observation}[theorem]{Observation}
\newtheorem{claim}{Claim}

\newtheorem{corollary}[theorem]{Corollary}
\newtheorem{remark}[theorem]{Remark}
\newtheorem{pro}[theorem]{Proposition}
\newenvironment {proof} {\noindent{\em Proof.}}{\hspace*{\fill}$\Box$\par\vspace{4mm}}
\newcommand{\ml}{l\kern-0.55mm\char39\kern-0.3mm}

\begin{document}
\title{\textbf{Conflict-free
chromatic index of trees}%\footnote{Supported by the National Science Foundation of China No. 12201375.}
}

\author{Shanshan Guo\footnote{Center for Discrete Mathematics and Theoretical Computer Science, Fuzhou University, Fuzhou, Fujian, China. {\tt
15738385820@163.com}}, \ \
Ethan Y.H. Li \footnote{School of Mathematics and Statistics,
Shaanxi Normal University, Xi'an, Shaanxi, China.
{\tt Email: yinhao\_li@snnu.edu.cn} }, \ \
Luyi Li \footnote{Academy of Mathematics and Systems Science, Chinese Academy of Sciences, Beijing, China.
{\tt Email: liluyiplus@gmail.com} }, \ \
Ping Li\footnote{Corresponding author: School of Mathematics and
Statistics, Shaanxi Normal University, Xi'an, Shaanxi, China. {\tt
Email: lp-math@snnu.edu.cn}}
}
\date{}
\maketitle

\begin{abstract}
A graph $G$ is  conflict-free $k$-edge-colorable
if there exists an assignment of $k$ colors to $E(G)$
such that
for every edge $e\in E(G)$,
there is a color that is assigned to exactly one edge among
the closed neighborhood of $e$.
The smallest $k$ such that $G$ is conflict-free $k$-edge-colorable
is called the conflict-free chromatic index of $G$, denoted
$\chi'_{CF}(G)$.
D\c{e}bski and Przyby\a{l}o showed that
$2\le\chi'_{CF}(T)\le 3$ for every tree $T$ of size at least two.
In this paper,
we present an algorithm to determine the conflict-free chromatic index of a tree without 2-degree vertices,
in time $O(|V(T)|)$.
This partially answer a question raised by Kamyczura, Meszka and Przyby\a{l}o.
\\[0.2cm]
\end{abstract}

\section{Introduction}

Motivated by frequency assignment in cellular networks, Even et al. \cite{Even} and Smorodinsky \cite{Smor} started studying  conflict-free vertex-coloring of graphs.
Let $G$ be a graph with vertex set $V(G)$ and edge set $E(G)$.
For every vertex $v\in V(G)$, let $N_G[v]= N_G(v)\cup\{v\}$.
If there is a vertex coloring $c:V(G)\rightarrow \mathbb{N}_+$ such that
for each vertex $v\in V(G)$,
there exists a vertex $w\in N_G[v]$  such that $c(w)$ is unique in $N_G[v]$
and the size of $c$ is as small as possible,
then  the size of $c$ is said to be the\emph{ conflict-free chromatic number} of $G$.
In the past twenty years,
the study of conflict-free chromatic number of graphs  has witnessed significant developments.
%Abel et al. \cite{Abel} presented
%that the conflict-free chromatic number of planar graph is $3$.
%Alon and Smorodinsky \cite{Alon} .
For more results, please refer to \cite{Abel,Alon,Even,Gargano,Horev,Lev,Pach,Smor}.

Recently, D\c{e}bski and  Przyby\a{l}o \cite{MJ2022} presented
an edge  version of conflict-free coloring.
Let $E_G(v)$
denote the set of edges incident with a vertex $v$ in $G$,
and let $E_G(uv):=E_G(u)\cup E_G(v)$  denote the \emph{closed neighbourhood} of
every edge $uv\in E(G)$.
When no confusion can occur, we shortly write $E(v)$ and $E(uv)$ respectively.
An \emph{edge-coloring} $c$ of  $G$ is a mapping
from $E(G)$ to a color set. In an edge-coloring $c$, if a color is assigned to exactly one edge in $E_G(e)$, then we call it a \emph{conflict-free color} of $e$. Note that an edge may have more than one conflict-free colors. A graph $G$ is called \emph{conflict-free $k$-edge-colorable} if there exists an edge-coloring of $k$ colors such that each edge $e\in E(G)$ has a conflict-free color.
The smallest $k$ that $G$ is conflict-free $k$-edge-colorable
is called the \emph{conflict-free chromatic index} of $G$,
denoted $\chi'_{CF}(G)$.
In addition,
D\c{e}bski and  Przyby\a{l}o \cite{MJ2022} also showed
that the conflict-free chromatic index of graph $G$ of maximum degree
$\Delta$ is at most $O(\ln \Delta)$
and the conflict-free chromatic index of $K_n$ is at least $\Omega(\ln n)$.
D\c{e}bski and  Przyby\a{l}o \cite{MJ2022}, and
Kamyczura et al. \cite{KMP}
gave the following result  independently.

\begin{theorem}[\cite{MJ2022,KMP}]
For any tree $T$, $\chi'_{CF}(T)\le 3$.
\end{theorem}

Note that the upper bound for the conflict-free chromatic index of a tree is tight, and it is reached when $T$ is a complete binary tree of height 3.
Furthermore, Kamyczura et al. \cite{KMP} raised the following problem.

\begin{problem}[\cite{KMP}]
  Characterize the family of all trees $T$ with $\chi'_{CF}(T)=3$.
\end{problem}

In this paper, we study the above problem by forbidding 2-degree vertices in $T$.
We now introduce some notations. In this paper we shall always assume that in any 2-edge-coloring of $T$ the edges are colored red or blue, and we use $E_r,E_b$ to denote the sets of edges with color red and blue, respectively. For a vertex $v$ of $T$, if all but one edge $e$ of $E(v)$ is colored by red (resp. blue), then red (resp. blue) is called the {\em unique color} on $E(v)$, and $e$ is called the {\em unique edge} of $E(v)$. The unique color and unique edge of $E(uv)$ are defined similarly.
For a rooted tree $T$,
we call each non-root vertex $u$  a {\em leaf (vertex)} if its degree $d_T(u)=1$, and call each vertex  of degree greater than one an {\em inner vertex}. Moreover, the edge incident with a leaf is called a {\em leaf edge}.
For any non-root vertex $v\in V(T)$,
we use $v^+$ to denote the father of $v$.

A rooted tree of level $\ell+1$ is called a {\em full tree} 
if the $0$-th level has exactly one vertex (the {\em root vertex} of $T$), and for each $1\leq i\leq \ell-1$, each vertex of the $i$-th level has at least two sons.
The level of a tree $T$ is denoted by $\ell(T)$ (note that if $T$ is an isolated vertex, then $\ell(T)=1$).
A rooted tree is a
{\em complete tree} if each inner vertex  has at least two sons.
Note that a full tree must be a complete tree. %Note that if  each inner vertex of a rooted tree $T$ has degree at least three, then $T$ is a complete tree.
It is obvious that full trees and complete trees do not contain 2-degree vertices.
Denote the vertex set in the $i$-th level of $T$ by $L_i(T)$.
%Define $v$-component to a connect component of  subgraph $T-v^+v$ which contains $v$.
For a (partial edge-colored) tree $T$ and a  vertex $v\in V(T)$,
we use $Sub_T(v)$ to denote the (partial edge-colored) subtree induced by $v^+ ,v$ and all descendants of $v$.
%the connect component of subgraph $T-v^{++}$ which contains $v$.
%$v^+$-component that contains $v$.
If $Sub_T(v)$ is a full tree but
$Sub_T(v^+)$ is not a full tree,
then we say $Sub_T(v)$ is a {\em maximal full subtree} of $T$ with root vertex $v^+$.

The rest of the paper is organized as follows.
In Section 2,  we give a sufficient and necessary condition for trees 
without 2-degree vertices being conflict-free 2-edge-colorable.
Section 3 is devoted to studying
the local construction of trees with conflict-free number two
without 2-degree vertices.
Using these constructions,  we presents an algorithm to determine the conflict-free chromatic index of trees without 2-degree vertices in time $O(|V(T)|)$, and we prove the feasibility of the algorithm.
In Section 4, we consider 2-degree vertices and give a sufficient condition for the trees with conflict-free index two.

\section{Characterizations of trees with conflict-free index two}
In this section, we give a sufficient and  necessary condition
for trees without $2$-degree vertices being conflict-free 2-edge-colorable.
We first give a simple observation as follows.

\begin{observation}\label{obs}
For a tree $T$, if $\chi'_{cf}(T)=2$ and $\gamma$ is a conflict-free red/blue edge-coloring of $T$, then for every inner vertex $v$, either $E(v)$ is monochromatic or $E(v)$ contains a unique color. Moreover, if $v$ is incident with a pendent edge, then $E(v)$ contains a unique color.
\end{observation}

\begin{lemma}\label{samecolor}
Let $T$ be a tree without $2$-degree vertices. If $\chi'_{cf}(T)=2$,
then for each conflict-free red/blue $2$-edge-coloring,
there is a color being the only conflict-free color of all edges in $E(T)$.
\end{lemma}

\begin{proof}
%Let $\chi'_{cf}(T)=2$.
%Select one edge $e = uv$ of $T$.
By Observation \ref{obs},
we may assume that there exist  an edge $e$ of $E(T)$ and a color,
say red,
such that $e$ is a red edge and
red is the conflict-free color of $e$.
It follows that all the edges in $E_T(e) \setminus \{e\}$ must be blue edges.
For $f \in E_T(e) \setminus \{e\}$, we have $d_T(V(f)\cap V(e))\ge 3$ since $T$ contains no 2-degree vertices, which yields $|E_T(f)\cap E_T(e)|\ge 3$. % for $f \in E_T(e) \setminus \{e\}$.
This implies that $E_T(f)$
contains exactly one red edge and at least two blue edges.
Thus, all the edges in $E_T(f) \setminus \{e\}$ must be blue edges and the conflict-free color of $f$ is red. Continuing this process, it follows that red is the only conflict-free color for each $e \in E(T)$.
Then the lemma holds.
\end{proof}
From now on we will call this color \emph{the conflict-free color} of $T$.

\begin{theorem}
Let $T$ be a tree of at least 3 vertices without 2-degree vertices. Then $\chi'_{cf}(T)=2$  if and only if $T$ has a maximal matching $M$ such that $T[V(M)]=M$.

\end{theorem}

\begin{proof}
Let $\chi'_{cf}(T)=2$ and take a conflict-free 2-edge-coloring of $T$.
By Lemma \ref{samecolor},
there exists a color, say red, being the conflict-free color of all edges
in $E(T)$. It follows that $E_r$ is a matching.
Then $T[V(E_r)] = E_r$ since otherwise there exists a blue edge connecting two red edges, which implies that red is not the conflict-free color of this edge,
a contradiction.
Suppose that $E_r$ is not maximal and there exists $g \in E(T)$ such that $E_r \cup \{g\}$ is a matching and $T[V(E_r \cup \{g\})] = E_r \cup \{g\}$.
Then $g$ is colored blue and all the edges adjacent to $g$ are colored blue, which is impossible since the edge-coloring of $T$ is conflict-free.

Conversely, if $T$ has a maximal matching $M$ such that $T[V(M)]=M$, then color the edges in $M$ red and color the edges in $E(T)-M$ blue. Suppose the resulting edge-coloring is not conflict-free. Then there must exist an edge $e$ such that $E_T(e)$ contains two red edges, which implies that $M$ is not a matching or $E(T[V(M)])- M \neq \emptyset$, a contradiction. It follows that $\chi'_{cf}(T)=2$ since $T$ has at least 2 edges.
\end{proof}

\section{Binary trees}

In this section,
all trees $T$ are oriented as out-branchings such that the degree of the root vertex is one, and for convenience, we call $T$ a tree instead of an out-branching.
If $\chi'_{cf}(T)=2$, then for any conflict-free edge-coloring of $T$ by two colors red and blue, and by Lemma \ref{samecolor} we may always assume that conflict-free color of $T$ is red.
It follows that for each inner vertex $v\in V(T)$, there is at most one red edge incident with $v$. If all out-edges of $v$ are blue, then we call $v$ an \textit{S-vertex}; if there is an out-edge of $v$ is red, then we call $v$ a \textit{D-vertex}, see Figure \ref{fig-SDvertex}.

\begin{figure}[htbp]
	\centering
	\scalebox{1}{\includegraphics[width=0.6\textwidth]{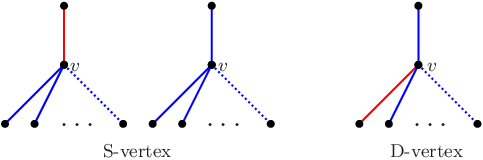}}\\
	\caption{$S$-vertex and $D$-vertex (red is the conflict-free color of $T$)}\label{fig-SDvertex}
\end{figure}

\begin{lemma}\label{same}
	Let $T$ be a full subtree of some tree $F$. In each conflict-free $2$-edge-coloring of $F$, the vertices in the same level of $T$ are either all S-vertices or all D-vertices.
\end{lemma}

\begin{proof}
Suppose to the contrary that there exists a conflict-free 2-edge-coloring for $F$ such that there are two vertices $v_1,v_2\in L_k(T)$ with $v_1$ being an S-vertex and $v_2$ being a D-vertex.
For our purpose, we may assume
$k$ is as large as possible. Recall that red is the conflict-free color of $F$.
Since $v_1$ is an S-vertex and $v_2$ is a D-vertex,
$v_2 v_2^+$ is blue, and there is an out-edge of $v_2$ is red and all out-edges of $v_1$ are blue.

If $v_1 v_1^+$ is red, then $v_1^+\neq v_2^+$ since otherwise $v_2 v_2^+$ has two adjacent red edges. Let $v_1'\neq v_1$ be a son of $v_1^+$.
Then all edges incident with $v_1'$ are blue since $v_1 v_1^+$ is red.
Hence, $v_1'$ is an $S$-vertex and each out-edge can not be a leaf edge (for otherwise this out-edge does not have a conflict-free edge, a contradiction).
It follows that $v_1$ also have two sons since $T$ is a full tree.
%$v_1$ is also a $S$-vertex.
Let $w,w'$ be sons of $v_1,v_1'$, respectively.
Note that $w'$ must be incident with a red out-edge.
Then $w,w'$ are not leaf-vertices and it is easy to verify that $w$ is an S-vertex and $w'$ is a D-vertex, contradicting the maximality of $k$.

If $v_1 v_1^+$ is blue, then all edges incident with $v_1$ are blue, and hence each out-edge of $v_1$ is not a leaf edge.
Since $T$ is a full tree and $v_2$ lies on the same level as $v_1$, each out-edge of $v_2$ is also not a leaf edge.
Then there exists a son $w$ of $v_2$ such that $v_2w$ is red since $v_2$ is a D-vertex. It follows that $w$ is not a leaf and hence $w$ is an S-vertex.
However, since all edges incident with $v_1$ are colored blue, it is easy to verify that each son of $v_1$ is a D-vertex, which contradicts the maximality of $k$.
\end{proof}

By Lemma \ref{same} we give the following definition.
\begin{definition}
Let $C$ be a conflict-free $2$-edge-coloring of a tree $F$ and $T$ be a full subtree of $F$. Denote by $C_T$ the restriction of $C$ on $T$. Then each vertex $u$ in $L_i(T)$ $(1 \le i \le \ell-1)$ is an $X_i$-vertex, where $X_i\in\{S,D\}$.
Define the coloring pattern of $C_T$ as
\begin{equation*}\label{c}
	cp(C_T)=\left\{
	\begin{array}{ll}
		(R,X_1,X_2,\ldots,X_{L(T)}),  &\mbox{if the root edge receives the conflict-free color red},\\
	
		(B,X_1,X_2,\ldots,X_{L(T)}),   &\mbox{otherwise}.
	\end{array} \right.
\end{equation*}
We also define the coloring pattern set of $T$ to be
\begin{equation*}
  cp(T) = \{cp(C_T) \colon \mbox{$C$ is a conflict-free $2$-edge-coloring of $F$}\}.
\end{equation*}
\end{definition}

Then we proceed to show that in any conflict-free 2-edge-coloring of a complete tree $T$, the coloring patterns for its maximal full subtrees are limited to several fixed cases. Recall that if $Sub_T(v)$ is a full tree but $Sub_T(v^+)$ is not a full tree, then we say $Sub_T(v)$ is a maximal full subtree of $T$ with root vertex $v^+$.

\begin{lemma}\label{subtree}
	Let $T$ be a complete tree and
	$T'$ be the maximal full subtree of $T$.
	If $T$ is conflict-free $2$-edge-colorable,
	then $2\le \ell(T')\le 5$ and  $cp(T') \subseteq \{(B),(R),R_1,R_2,R_3,R_4\}$,
	where
	$R_1=(B,D)$,
	$R_2=(R,S)$, $R_3=(R,S,S,D)$ and $R_4=(B,S,D)$ as shown in Figure \ref{fig-cp}.
\end{lemma}

\begin{figure}[htbp]
	\centering
	\scalebox{0.75}{\includegraphics[width=0.8\textwidth]{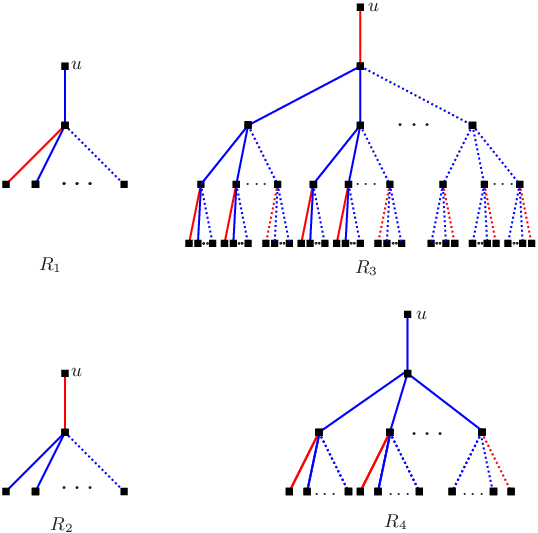}}\\
	\caption{Coloring patterns for maximal full subtrees}\label{fig-cp}
\end{figure}

\begin{proof}
Let $v^+$ be the root of $T'$ and $vv^+$ be the edge of $T'$ incident with $v$.
If $\ell(T')=2$, then $cp(T')\subseteq \{(B),(R)\}$. Now let $\ell(T') \ge 3$. Then we have the following  three cases to consider.

{\bf Case 1.}
$vv^+$ is red.

In this case $vw$ is blue for each son $w$ of $v$.
If $\ell(T')=3$, then $T'\cong R_2$.
If $4\le \ell(T')\le 5$,
then for each son $w$ of $v$, all edges incident with $w$ are blue.
Hence each son of $w$ is not a leaf and $\ell(T') = 5$.
Since $T'$ is a full tree, all sons of $w$ are D-vertices and $cp(T') = R_3$.
If $\ell(T')>5$, then since each vertex $x$ of $L_3(T')$ is a D-vertex, there are two sons of $x$, say $y_1,y_2$, such that $xy_1$ is red and  $xy_2$ is blue.
Then all edges incident with $y_2$ are blue edges, and hence each son of $y_2$ is not a leaf.
Since $T'$ is a full tree, each son of $y_1$ is also not a leaf.
Let $z_1,z_2$ be sons of $y_1,y_2$, respectively.
Then $z_1$ is an S-vertex and $z_2$ is a D-vertex, which contradicts Lemma \ref{same} since $z_1,z_2\in L_5(T')$.

{\bf Case 2.} All edges incident with $v$ are colored blue.

In this case each son $w$ of $v$ is a D-vertex and there exists a son $x$ of $w$ such that $wx$ is colored red.
Hence, if $\ell(T')=4$, then $cp(T') = R_4$.
If $\ell(T')>4$, then for any son $y\neq x$ of $w$, all edges incident with $y$ are blue.
Hence, any son of $y$ is not a leaf and each son $y'$ of $y$ is a D-vertex.
Since $x,y\in L_3(x)$ and $T'$ is a full tree, $x$ has a son $x'$ and $x'$ must be an S-vertex. This leads to a contradiction as $x'$ and $y'$ lies in the same level.

{\bf Case 3.} There is a son $w$ of $v$ with $vw$ colored red.

In this case $vv^+$ is blue, and for each son $w' \neq w$ of $v$ the edge $vw'$ is also blue.
If $\ell(T') = 3$, then $cp(T') = R_1$.
Now we assume that $\ell(T') \ge 4$.
Since $Sub_T(w)$ is a full tree, as discussed in Case 1, $Sub_T(w)$ is either $R_2$ or $R_3$.
For any son $w'$ of $v$ with $w'\neq w$, all edges incident with $w'$ are colored blue.
As discussed in Case 2, $cp(Sub_T(w')) = R_4$.
However, the level of $Sub_T(w)$ is either 3 or 5, and the level of $Sub_T(w')$
is 4. This implies that $T'$ is not a full tree, a contradiction.
\end{proof}

Now, we give some definitions  and new graphs to
further discuss the construction of complete  trees by Lemma \ref{subtree}.
Let $T_1,T_2,\ldots, T_k$ be  complete trees and
$v_i$ be the root of $T_i$ for each $i\in [k].$
We construct a complete tree $T$ from $T_1,T_2,\ldots T_k$ and a new edge $uv$
by identifying  $v_1,v_2\ldots v_k$ and $v$, denoted by
$T:=T_1\bigoplus T_2\bigoplus\ldots \bigoplus T_k$, see Figure \ref{sum}.

\begin{figure}[ht]
	\centering
	\includegraphics[width=340pt]{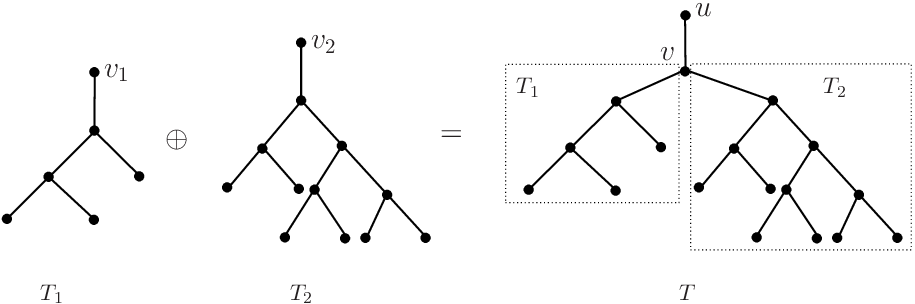}\\
	\caption{The sum of some trees}\label{sum}
\end{figure}

Let $T$ be a conflict-free 2-edge-colorable tree.
A vertex $v\in V(T)$ is called a  {\em fixed vertex}
if the coloring pattern of $T[E(v)]$ is the same
for each conflict-free 2-edge-coloring of $T$.
Let  $I=\{x^+\in V(T): Sub_T(x)\mbox{ is a maximal full tree}\}$.
A vertex $u\in L_i(T)$ is a {\em surficial vertex} of $T$ if  $i$ is maximum
in $I$ (in other words, the surficial vertex is a vertex in $I$ with largest level).

\begin{pro}\label{sur}
For each son $v$ of a surficial vertex $u$ in $T$, $Sub_T(v)$ is a  full tree.
\end{pro}
\begin{proof}
Suppose to the contrary that
there exists a son $v'$ of $u$ such that $Sub_T(v')$ is not a full tree.
Then $Sub_T(v')$ has a maximal full tree, say $Sub_{Sub_T(v')}(w)=Sub_{T}(w)$.
This implies that $w^+\in I$.
Thus, we have $\ell(w^+)\ge \ell(v')>\ell(u)$, which contradicts the maximality of $u$.
\end{proof}

Let $\mathcal{T}_k$ be a set of full trees $T$ with $\ell(T)=k$.
Let $\mathcal{T}^i_k$ denote the sum of
a family of $i$ (not necessarily distinct) elements from $\mathcal{T}_k$.
Now we define four tree families as follows.
\begin{itemize}
\item $\mathcal{F}_1:=\{\mathcal{T}^{k_1}_2
   \bigoplus \mathcal{T}^{1}_3:  k_1> 0  \}$,

\item $\mathcal{F}_2:=\{\mathcal{T}^{k_2}_2\bigoplus
   \mathcal{T}^{k_3}_4: k_2,k_3> 0  \}$,

\item $\mathcal{F}_3:=
   \{\mathcal{T}^{k_4}_2\bigoplus \mathcal{T}^{1}_3
   \bigoplus \mathcal{T}^{k_5}_4   : k_5> 0 \}$,

\item $\mathcal{F}_4:=
   \{\mathcal{T}^{k_6}_2\bigoplus \mathcal{T}^{k_7}_4
   \bigoplus \mathcal{T}^{1}_5:k_6+k_7> 0 \}$.
     \end{itemize}
It is clear that each element of $\mathcal{F}_1,\mathcal{F}_2,\mathcal{F}_3,\mathcal{F}_4$ is not a full tree. Moreover, in any conflict-free 2-edge-coloring of such a tree $T$, since $T$ consists of trees in $\mathcal{T}_1\cup \mathcal{T}_2\cup \mathcal{T}_3\cup \mathcal{T}_4$,  the coloring pattern of $T$ is determined or partially determined by Lemma \ref{subtree}.
In fact, each conflict-free 2-edge-coloring of each element in $\mathcal{F}_1,\mathcal{F}_3,\mathcal{F}_4$ is determined by Lemma \ref{subtree}, and we use $\mathcal{F}_i^*$ to denote the set of $F\in\mathcal{F}_i$ associated with the unique 2-edge-coloring (see Figure \ref{three-figs}, the red/blue edge-coloring of $F_1\in \mathcal{F}_1^*,F_3\in \mathcal{F}_3^*$ and $F_4\in \mathcal{F}_4^*$ are determined).
For a tree $\mathcal{T}^{k_2}_2\bigoplus\mathcal{T}^{k_3}_4$ of $\mathcal{F}_2$,
the coloring pattern of each $\mathcal{T}^{k_3}_4$ is determined.
We define $\mathcal{F}_2^*$ as the set of $\mathcal{T}^{k_2}_2\bigoplus\mathcal{T}^{k_3}_4\in\mathcal{F}_2$ associated with the colors on
$$\bigcup\{E(z):z\mbox{ is the a vertex in penultimate level of }\mathcal{T}^{k_3}_4\},$$
such that the color pattern of each $E(z)$ is the same as the color pattern in the unique 2-edge-coloring of the $\mathcal{T}^{k_3}_4$ (see Figure \ref{three-figs}, $F_2\in \mathcal{F}_2^*$ is a partial red/blue edge-coloring. Note that the black edges in $E(u)$ are uncolored edges).

For each element in $\mathcal{F}_2$, the corresponding partial 2-edge-coloring in $\mathcal{F}^*_2$ can be extended to two conflict-free 2-edge-coloring patterns $\mathcal{F}_2^1$ and $\mathcal{F}_2^2$ (see Figure \ref{two-f7}, the two conflict-free 2-edge-colorings are  distinguished by the types of $v$, say D-vertex or S-vertex).

\begin{theorem}\label{111111}
Let $T$ be a complete tree but not a full tree,
and let $u$ be a surficial vertex of $T$.
If $T$ is conflict-free $2$-edge-colorable, then $Sub_T(u)\in\mathcal{F}_1\cup \mathcal{F}_2\cup\mathcal{F}_3\cup\mathcal{F}_4$.
In addition, the following statements hold.
\begin{enumerate}
  \item [(1)] If $Sub_T(u)$ belongs to $\mathcal{F}_1,\mathcal{F}_3$ or $\mathcal{F}_4$, then the edge-colorings are presented as in Figure \ref{three-figs}, respectively. Moreover, $u$ is a fixed vertex.
  \item [(2)] If $Sub_T(u)\in \mathcal{F}_2$, then the partial edge-coloring of $Sub_T(u)$ can be extended to more levels, as shown in Figure \ref{two-f7}.
\end{enumerate}
\end{theorem}

\begin{proof}
Since $u$ is a surficial vertex of $T$,
$Sub_T(v^\star)$ is a maximal full  tree for each son $v^\star$ of $u$ by Proposition \ref{sur} and $Sub_T(u)$ is not a full tree.
Then there are two sons $v,v'$ of $u$
such that $\ell(Sub_T(v))\neq \ell(Sub_T(v'))$.
Without loss of generality, we assume that $\ell(Sub_T(v))$ is maximum and  $\ell(Sub_T(v'))$ is minimum among all sons of $u$.
By Lemma  \ref{subtree},
we have $2\le \ell(Sub_T(v'))<\ell(Sub_T(v))\le 5$ (this indicates $3\leq \ell(Sub_T(v))\leq 5)$.
Recall that red is the conflict-free color of $T$.
We have the following three cases to discuss.

{\bf Case 1.} $\ell(Sub_T(v))=3$.

In this case  $\ell(Sub_T(v'))=2$, that is, $v'$ is a leaf-vertex.
If $cp(Sub_T(v))= R_1$, then the colors of all edges incident with $u$ are blue, which implies that $uv'$ does not have a conflict-free edge, a contradiction.
If $cp(Sub_T(v))= R_2$,
then
$\ell(Sub_T(v^\star))=2$ and $uv^\star$ is blue for each son $v^\star\neq v$ of $u$.
Thus, there exists an integer $k_1>0$ such that
$Sub_T(u)= \mathcal{T}^{k_1}_2
   \bigoplus \mathcal{T}^{1}_3 \in \mathcal{F}_1$.
Moreover, $u$ is a fixed vertex since the edge-coloring of $Sub_T(u)$ is fixed.

\begin{figure}[htbp]
    \centering
    \includegraphics[scale=0.6,trim={1cm 0.7cm 1cm 0.4cm}]{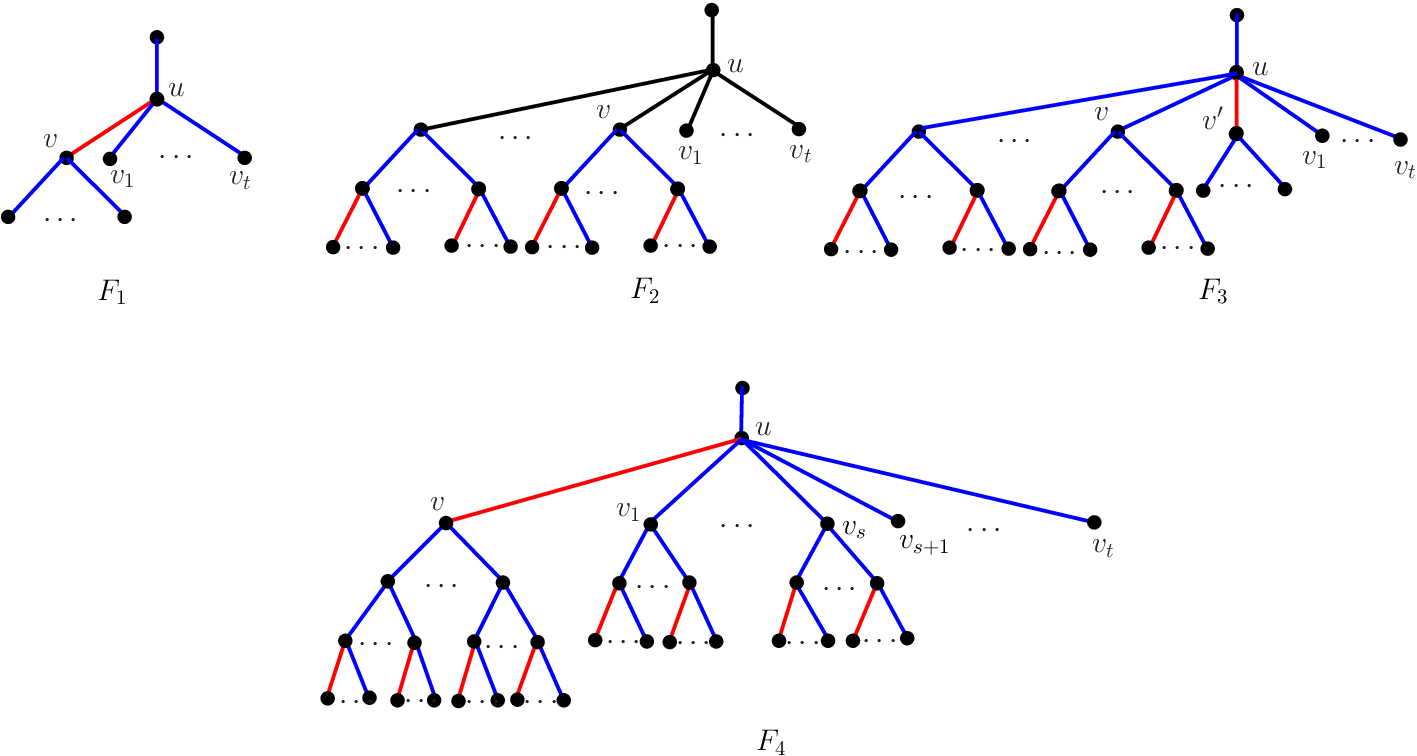}
    \caption{Partial edge-colorings of $\mathcal{F}_1,\mathcal{F}_2,\mathcal{F}_3,\mathcal{F}_4$} \label{three-figs}
\end{figure}

\begin{figure}[htbp]
    \centering
    \includegraphics[scale=0.65,trim={1cm 0.5cm 1cm 0.3cm}]{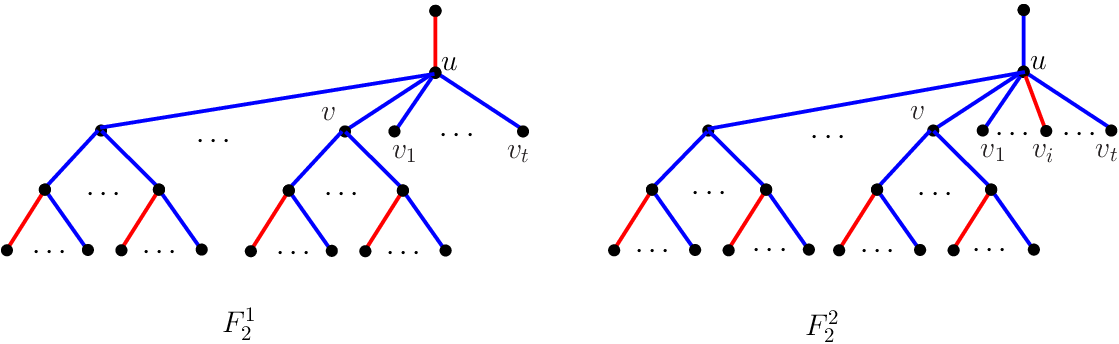}\\
    \caption{Partial edge-colorings of $\mathcal{F}_2$ (extended)} \label{two-f7}
\end{figure}

{\bf Case 2.} $\ell(Sub_T(v))=4$.

In this case $\ell(Sub_T(v'))\in\{2,3\}$, $cp(Sub_T(v))=R_4$
and there is a red edge $f$ incident with $u$ in $T$.
If  $\ell(Sub_T(v'))=3$,
then  $cp(Sub_T(v'))\subseteq \{R_1,R_2\}$ by Lemma \ref{subtree}.
If $cp(Sub_T(v'))= R_1$,
then $uv'$ is incident with two distinct red edges,
a contradiction.
If $cp(Sub_T(v'))= R_2$,
then $f=uv'$ and there exists an integer $k_5>0$ such that
$Sub_T(u)= \mathcal{T}^{1}_3
   \bigoplus \mathcal{T}^{k_5}_4 \in \mathcal{F}_3$.
Moreover, $u$ is a fixed vertex since the edge-coloring of  $Sub_T(v')$ is fixed.

Now we consider the case $\ell(Sub_T(v'))=2$.
If the color of $uv'$ is red, then $uv^\star$ is blue for all sons $v^\star \neq v'$ of $u$. Further, if $\ell(Sub_T(v^\star))=3$, then $uv^\star$ does not have a conflict-free color whenever $cp(Sub_T(v^\star)) = R_2$ or $cp(Sub_T(v^\star)) = R_1$, a contradiction. Hence $Sub_T(u)=\mathcal{T}^{k_2}_2
   \bigoplus \mathcal{T}^{k_3}_4 \in \mathcal{F}_2$ for some integers $k_2>0$ and $k_3>0$, and its partial edge-coloring coincides with $F_2^2$.
If the color of $uv'$ is blue,
then we may assume that
the color of $uv''$ is blue
for each son $v''$ of $u$ with
$\ell(Sub_T(v'))=2$, for otherwise we could replace $v'$ by $v''$ and apply the previous arguments.
Then $Sub_T(u)=\mathcal{T}^{k_2}_2\bigoplus \mathcal{T}^{k_1}_4  \in \mathcal{F}_2$ with partial edge-coloring being $F_2^1$ or there exists a son $v'''$ of $u$ with
$\ell(Sub_T(v'''))=3$ and $cp(Sub_T(v'''))=R_2$, which implies that
$Sub_T(u)=\mathcal{T}^{k_4}_2
   \bigoplus \mathcal{T}^{1}_3 \bigoplus \mathcal{T}^{k_5}_4  \in \mathcal{F}_3$ with $k_4 \ge 0$ and $k_5 \ge 1$.
Note that $u$ is not a fixed vertex in $\mathcal{F}_2$ but is fixed in $\mathcal{F}_3$.

{\bf Case 3.} $\ell(Sub_T(v))=5$.

In this case $\ell(Sub_T(v'))\in\{2,3,4\}$. We claim that there is no son $v^\star$ of $u$ such that $\ell(Sub_T(v^\star)) = 3$. Indeed, if such $v^\star$ exists,
then $cp(Sub_T(v^\star))\in \{R_1,R_2\} $ and there
exists a red edge incident with $v^\star$. This implies that there are two red edges incident with $uv^\star$, leading to a contradiction.
Hence there exist two integers $k_6,k_7$ such that
$Sub_T(u)=\mathcal{T}^{k_6}_2 \bigoplus \mathcal{T}^{k_7}_4 \bigoplus \mathcal{T}^{1}_5  \in \mathcal{F}_4$ and $k_6+k_7>0$.
Moreover, $u$ is a fixed vertex since the edge-colorings of $Sub_T(v)$ and $Sub_T(v')$ are fixed.
\end{proof}

\begin{corollary}\label{corr}
  Let $T$ be a complete tree but not a full tree, and let $u$ be a surficial vertex of $T$.
If $Sub_T(u)\notin \mathcal{F}_1\cup \mathcal{F}_2\cup\mathcal{F}_3\cup\mathcal{F}_4$, then $\chi'_{cf}(T)=3$.
\end{corollary}

Suppose that $T$ has a partial edge-coloring $\gamma$ on $E'\subseteq E(T)$.
We say that $\gamma$ can be {\em extended to a conflict-free $2$-edge-coloring} if there is a conflict-free 2-edge-coloring $\Gamma$ of $T$ such that $\gamma(e)=\Gamma(e)$ for each $e\in E'$.
Recall that if $T$ is a partially edge-colored tree and $u$ is an inner vertex of $T$, then $Sub_T(u)$ is a partial edge-colored subgraph inheriting the partial edge-coloring of $T$.
Let $(T,u)$ denote the partially edge-colored subtree that is obtained from $T$ by deleting all descendants of  all sons of $u$.
Note that $(T,u)$ is also a tree without 2-degree vertices.
Algorithm \ref{algorithm} gives an algorithm for determining $\chi'_{cf}(T)$, where $T$ is a tree without 2-degree vertices.
We  prove the feasibility  and discuss the complexity of Algorithm \ref{algorithm} in the following theorem.

\begin{algorithm}[htb!]
%\footnotesize
\small
\caption{Decide the conflict-free index of a tree  without $2$-degree}\label{algo}% 算法名字
\label{algorithm}
\LinesNumbered %要求显示行号
\KwIn{a tree $T$ without $2$-degree vertices.}% 输入参数
\KwOut{$\chi'_{cf}(T)=2$ or $\chi'_{cf}(T)=3$.}%输出
$G=T$\;
$U=E(G)$\;
%$i=1$\;
choose a leaf vertex $r$ of $G$, and
orient edges such that $G$ is an out-branching with root $r$\;
\While{$U\neq \emptyset$}{
choose a surficial vertex $u$\;
\If{$Sub_G(u)\vdash F$ for some $F\in\mathcal{F}^*_1\cup \mathcal{F}^*_3\cup \mathcal{F}^*_4$}{
        color $Sub_G(u)$ as in $F$\;
        $G=(G,u)$\;
        $U=U\cap E(G)-E(u)$\;}
       % assign $u$ the color ``green''\;
%        $s(u)=i$\;

\ElseIf{$Sub_G(u)\vdash F$ for some $F\in\mathcal{F}^*_2$}{
$G=(G,u)$\;
$U=U\cap E(G)$\;}

\Else
{output ``$\chi'_{cf}(T)=3$''\;
return\;}
$i=i+1$\;}

\If{the edge-coloring of $G$ is a conflict-free edge-coloring}{
        output ``$\chi'_{cf}(T)=2$''\;}
\Else
{output ``$\chi'_{cf}(T)=3$''\;}
\end{algorithm}

\begin{theorem}
Suppose that $T$ is a tree without $2$-degree vertices.
We can decide $\chi'_{cf}(T)$ by using Algorithm \ref{algorithm} in $O(|V(T)|)$ times.
\end{theorem}
\begin{proof}
Let $G_0=T$.
Suppose that  the ``while'' loop terminates after $n$ steps, and
after the $i$-th step of ``while'', the resulting partially edge-colored tree $G$ is denoted by $G_i$.
For the sake of discussion, we label the surficial vertex $u$ of $G_{i-1}$ as $s(u)=i$ (note that $G_i$ is obtained from $G_{i-1}$ by deleting all descendants but sons of $u$), and then assign $u$ the color green.
In the $i$-th step of ``while'', we must delete some edges of $G_{i-1}$ and then assign colors to an edge subset $E'$ of $G_i$.
Specifically, if $E'\neq \emptyset$, then $E'=E(u)$ and $G_i$ is obtained in lines 6--10 of Algorithm \ref{algorithm}; if $E'=\emptyset$, then $G_i$ is obtained in lines 11--15 of Algorithm \ref{algorithm}.
Note that in each step of ``while'', $Sub_G(u)$ is a partially edge-colored graph.
For easy of discussion, we use $Sub_G^{\downarrow}(u)$ to denote the graph obtained from $Sub_G(u)$ by removing all colors.

At first we prove  the feasibility of Algorithm \ref{algorithm}.

\begin{claim}\label{clmm-222}
If $G_i$ is obtained in lines 11--15 of Algorithm \ref{algorithm} and $u$ is a  green vertex in $G_i$ with $s(u)=i$, then  $E(u)$ is uncolored in $G_i$.
\end{claim}
\begin{proof}
Suppose to the contrary that there exists an edge  $e=uu'$ such that $e$ is colored.
Then $u'$ is a green vertex with $s(u')=j$ for some $j<i$.
If $u'=u^+$, then $u$ is a leaf vertex in $G_{j+1}$. Since $u\in V(G_i)$ and $G_i$ is a subtree of $G_{j+1}$, it follows that $u$ is also a leaf vertex in $G_i$.
This contradicts the the fact that  $Sub^{\downarrow}_{G_{i-1}}(u)\in\mathcal{F}_2$.
If $u'\neq u^+$, then $u'$ is a son of $u$, which implies that all descendants but sons of $u'$ are deleted.
We can get a contradiction by a similar way.
Therefore, $E(u)$ is uncolored in $G_i$.
\end{proof}

For convenience, we also regard an uncolored graph as a partially edge-colored graph.
For each integer $i\in [n]$, let $\gamma_i$ be a partial edge-coloring of $G_i$.

\begin{claim}\label{clmm-2}
For $i\in[n]$, $\gamma_i$ can be extended to a conflict-free 2-edge-coloring  of $G_i$
if and only if $\gamma_{i-1}$ can be extended to a conflict-free 2-edge-coloring of $G_{i-1}$.\end{claim}
\begin{proof}
Suppose that $G_i=(G_{i-1},u)$, i.e., $G_i$ is the graph obtained from $G_{i-1}$ by deleting all descendants but sons of $u$.
By lines 6--15 of Algorithm \ref{algorithm}, $Sub_G(u)\vdash F$ for some $F\in\mathcal{F}^*_1\cup\mathcal{F}^*_2\cup \mathcal{F}^*_3\cup \mathcal{F}^*_4$.
We consider the following two cases.

{\bf Case 1} $Sub^{\downarrow}_{G_{i-1}}(u)$ is a graph of $\mathcal{F}_j$, where $j\in\{1,3,4\}$.

We first prove the sufficiency.
If $\gamma_{i-1}$ can be extended to a conflict-free 2-edge-coloring $\Gamma_{i-1}$ of $G_{i-1}$, then there is a red edge incident with $u$ by Theorem
\ref{111111}.
Hence,
$\Gamma_{i-1}|_{G_i}$ is a conflict-free 2-edge-coloring of $G_i$.
Next, we only need to show that $\Gamma_{i-1}|_{G_i}$ is an edge-coloring extended from  $\gamma_i$, that is, to show that for each red (resp. blue) edge $e\in E(G_i)$ under $\gamma_i$, $e$ is also a red (resp. blue) edge under $\Gamma_{i-1}|_{G_i}$. If $e\notin E(u)$, then since $\gamma_i$ is obtained from $\gamma_{i-1}$ by coloring only edges incident with the green vertices in $G_i$,
it follows that $e$ is red (resp. blue) under $\gamma_{i-1}$, and hence $e$ is also red (resp. blue) under $\Gamma_{i-1}|_{G_i}$.
If $e\in E(u)$, then since $u$ is a fixed vertex by Theorem \ref{111111}, the color pattern of $E(u)$ in $\gamma_i$ is the same as in $\gamma_{i-1}$, and also the same as in $\Gamma_{i-1}|_{G_i}$.

Now we proceed to prove the necessity. Assume that $\gamma_i$ can be extended to a conflict-free 2-edge-coloring $\Gamma_i$ of $G_i$. Since $u$ is incident with a leaf vertex in $G_i$,
it follows that there is a red edge incident with $u$.
Since $Sub_G(u)\vdash F$ for some $F\in\mathcal{F}^*_1\cup  \mathcal{F}^*_3\cup \mathcal{F}^*_4$, the union
 of $\Gamma_i$ and the edge-coloring of $F$, denoted by $\Gamma^*$, is a conflict-free 2-edge-coloring of $G_{i-1}$.
Note that $\gamma_{i}$ is obtained from $\gamma_{i-1}|_{G_i}$ and the edge-coloring of $E(u)$ as $Sub_G(u)$.
Thus, $\gamma_{i-1}$ can be extended to $\Gamma^*$.

{\bf Case 2} $Sub^{\downarrow}_{G_{i-1}}(u)$ is a graph of $\mathcal{F}_2$.

If $\gamma_{i-1}$ can be extended to a conflict-free 2-edge-coloring $\Gamma_{i-1}$ of $G_{i-1}$, then there is a red edge  incident with $u$ whenever $Sub_{G_{i-1}}(u)$ is a graph of $\mathcal{F}_2^1$ or $\mathcal{F}_2^2$. Hence, $\Gamma_{i-1}|_{G_i}$ is a conflict-free 2-edge-coloring of $G_i$.
In order to prove the sufficiency, we only need to show that $\Gamma_{i-1}|_{G_i}$ is an edge-coloring extended from  $\gamma_i$, that is, to show that for each red (resp. blue) edge $e\in E(G_i)$ under $\gamma_i$, $e$ is also a red (resp. blue) edge under $\Gamma_{i-1}|_{G_i}$. By Claim \ref{clmm-222}, each of $E(u)$ is uncolored in $G_i$.
Hence $e\notin E(u)$.
Since $\gamma_i$ is obtained from $\gamma_{i-1}$ by coloring only edges incident with the green vertices in $G_i$,
it follows that $e$ is red (resp. blue) under $\gamma_{i-1}$, and hence $e$ is also red (resp. blue) $\Gamma_{i-1}|_{G_i}$.

Then we proceed to show the necessity. Assume that $\gamma_i$ can be extended to a conflict-free 2-edge-coloring $\Gamma_i$ of $G_i$.
Since $u$ is adjacent to a leaf vertex  in $G_i$,
it follows that there is a red edge incident with $u$,  see Figure \ref{two-f7}.
In either case, we can extend $\Gamma_{i}$ to a conflict-free $2$-edge-coloring of $G_{i-1}$.
Furthermore, this edge-coloring is also extended from $\gamma_{i-1}$.
\end{proof}

\begin{claim}\label{clmm-3}
Let $G_{i+1}=(G_i,u)$. If $Sub^{\downarrow}_{G_i}(u)$ is isomorphic to some graph in $\mathcal{F}_1\cup\mathcal{F}_2\cup\mathcal{F}_3\cup\mathcal{F}_4$ but  $Sub_{G_i}(u)\nvdash\mathcal{F}_1^*\cup\mathcal{F}^*_2\cup\mathcal{F}^*_3
\cup\mathcal{F}^*_4$, then $\chi'_{cf}(T)=3$.
\end{claim}
\begin{proof}
Suppose to the contrary that $\chi'_{cf}(T)=2$.
By Claim \ref{clmm-2}, we have that $2=\chi'_{cf}(T)=\chi'_{cf}(G_0)=\chi'_{cf}(G_1)=\cdots=\chi'_{cf}(G_{i})$, and the partial edge-coloring of $G_{i}$ can be extended to a conflict-free 2-edge-coloring of $G_{i}$.

If $Sub^{\downarrow}_{G_i}(u)$ is isomorphic to some element of $\mathcal{F}_1\cup\mathcal{F}_3\cup\mathcal{F}_4$, then $Sub_{G_i}(u)$ has the unique 2-edge-coloring in any conflict-free 2-edge-coloring of $G_{i}$ and the coloring pattern is the same as the corresponding element in $\mathcal{F}_1^*\cup\mathcal{F}^*_3
\cup\mathcal{F}^*_4$. Hence, $Sub_{G_i}(u)\vdash\mathcal{F}_1^*\cup\mathcal{F}^*_3\cup\mathcal{F}^*_4$, a contradiction.

If $Sub^{\downarrow}_{G_i}(u)$ is isomorphic to an element of $\mathcal{F}_2$, then $E(u)$ is uncolored in $Sub_{G_i}(u)$ by Claim \ref{clmm-222}. Note that in any conflict-free 2-edge-coloring $G_i$, the pattern of $Sub_{G_i}(u)$ belongs to $\mathcal{F}^1_2$ or $\mathcal{F}_2^2$.
Hence, $Sub_{G_i}(u)\vdash\mathcal{F}^*_2$, a contradiction.
Thus, $\chi'_{cf}(T)=3$.
\end{proof}

By Claim \ref{clmm-2}, the partial edge-coloring of $G_i$ can be extended to a conflict-free 2-edge-coloring if and only if $G_0=T$ has a conflict-free 2-edge-coloring for each $i\in[n]$.
Recall that the ``while'' stops after $n$ steps.
If the ``while'' stops when
  $Sub_{G_n}(u)$ does not belong to $\{F_1,F_2,F_3,F_4\}$, then $\chi'_{cf}(G_n)=3$ by Corollary \ref{corr}.
If the ``while'' stops when $Sub_{G_n}(u)$ is isomorphic to one graph of $\mathcal{F}_1\cup\mathcal{F}_2\cup \mathcal{F}_3\cup \mathcal{F}_4$ but the partial edge-coloring of $Sub_{G_n}(u)$ does not coincide with any edge-colored graph of $\mathcal{F}_1\cup\mathcal{F}_2\cup \mathcal{F}_3\cup \mathcal{F}_4$, then
$\chi'_{cf}(G_n)=3$ by Claim \ref{clmm-3}.
If the ``while'' loop terminates when $U=\emptyset$, then we get an edge-coloring of $G_n$.
By Claim \ref{clmm-2}, $\chi'_{cf}(G_n)=2$ if and only if $\chi'_{cf}(T)=2$. The proof is completed.

Next, we discuss the complexity of Algorithm \ref{algorithm}.
Recall that the tree $T$ is rooted at $r$ ($r$ is a leaf vertex). We label each vertex $v\in V(T)$ as $d_T(v,r)$, this takes $O(|V(T)|)$ times. 
Note that in the $i$-th step of Algorithm \ref{algorithm}, the subtree $G_i$ is also rooted at $r$ and each vertex $v\in V(G_i)$ is labelled by $d_{G_i}(v,r)=d_T(v,r)$.
In line 5 of Algorithm \ref{algorithm}, we use Algorithm \ref{alg} to find a surficial vertex $u$.
It is clear that Algorithm \ref{alg} can find a surficial vertex, since we begin with a vertex $x$ such that $d_T(r,x)$ is maximum.
%, which indicate that $Sub_T(x^+)$ is a full tree.
 
Assume that $u_i$ is the new surficial vertex in $G_i$ for each $0\leq i<n$.
Then $G_{i+1}$ is obtained from $G_i$ by deleting all descendants but sons of $u_i$.
It takes totally $O(\sum_{0\leq i<n}|Sub_{G_i}(u)|)$ times in line 5 of Algorithm \ref{algorithm}. Furthermore,  the ``while'' loop takes $O(\sum_{0\leq i<n}|Sub_{G_i}(u)|)$ times.
It is obvious that lines 21--26 of Algorithm \ref{algorithm} take $O(|V(T)|)$ times.
So, Algorithm \ref{algorithm} takes $O(|V(T)|)+O(|V(T)|)+O(\sum_{0\leq i<n}|Sub_{G_i}(u)|)=O(|V(T)|)$ times since $\sum_{0\leq i<n}|Sub_{G_i}(u)|\leq O(|V(T)|)$.
\end{proof}
\begin{algorithm}[htb!]
%\footnotesize
\small
\caption{Find a surficial vertex}\label{algo}% 算法名字
\label{alg}
\LinesNumbered %要求显示行号
\KwIn{a complete tree $T$ rooted at a leaf vertex $r$, with each vertex $u\in V(T)$ labelled by $\ell(v)=d_T(v,r)$.}% 输入参数
\KwOut{a surficial vertex $u$.}%输出
choose a leaf vertex $x$ with $\ell(x)$ maximum\;
let $u=x^+$\;
\While{$Sub_T(u)$ is a full tree}{
$u=u^+$\;}
\end{algorithm}

\section{Trees with 2-degree vertices}

Algorithm \ref{algorithm} can only distinguish $\chi'_{cf}(T)$ when $T$ is a tree without 2-degree vertices.
If $T$ has 2-degree vertices, then the problem is complicated since the conflict-free colors of the edges may not be the same and we cannot apply Lemma \ref{samecolor}.
Next, we give a sufficient condition for $\chi'_{cf}(T)=2$, where $T$ is a general tree. Let $T_{= 2}$ and $T_{\geq 3}$ denote subgraphs of $T$ induced by edge sets
$\bigcup_{v:d_T(v)= 2}E_T(v)$ and
$\bigcup_{v:d_T(v)\geq 3}E_T(v)$, respectively.
%For an edge $e\in E(T)$, it is clear that $e\in E(T_{= 2})\cap E(T_{\geq 3})$ if and only if one end vertex of $e$ is a 2-degree vertex, and the other end vertex of $e$ has degree at least three.

\begin{theorem}\label{thm-main-2}
For a tree $T$, if each component of $T_{\geq 3}$ is conflict-free $2$-edge-colorable and each component of $T_{=2}$ has at least $5$ vertices, then $\chi'_{cf}(T)=2$.
\end{theorem}

\begin{proof}
We prove the theorem by induction on $|T|$.
It is obvious that the result holds for $|T|\leq 4$.
If $T$ does not contain 2-degree vertices, then $\chi'_{cf}(T) = \chi'_{cf}(T_{\ge 3}) =2$.
So, assume that $T_{= 2}$ is a nonempty graph and $P=x_1x_2\ldots x_t$ is a component of $T_{= 2}$, where $t\geq 5$.
Let $P'=x_3x_4\ldots x_{t-2}$
and let $T_1,T_2$ be the two components of $T-V(P')$ such that $x_2$ is a leaf-vertex of $T_1$ and $x_{t-1}$ is a leaf-vertex of $T_2$. Then $T_1$ and $T_2$ are both conflict-free 2-edge colorable by induction.
If  $d_{T_1}(x_1)=1$, it follows that $T_1$ is an edge $x_2x_1$. This case is trivial since we can get a conflict-free red/blue edge-coloring of $T$ obtained from a conflict-free red/blue edge-coloring of $T_2$ by  coloring edges in $P-x_1$ with red and blue alternately.
Similarly the case $d_{T_2} (x_t) = 1$ is also trivial, and hence in the following we may assume that $E_{T_1}(x_1)$ and $E_{T_2}(x_t)$ have a conflict-free edge, respectively.
In order to show the theorem, we consider the following three cases.

\noindent{\bf Case 1.}
$x_1x_2$ and $x_{t-1} x_t$ are the conflict-free edges of $E_{T_1}(x_1x_2)$ and $E_{T_2}(x_{t-1}x_{t})$, respectively.

Note that we can give conflict-free edge-colorings to $T_1$ and $T_2$ such that  the colors of $x_1x_2$ and $x_{t-1} x_t$ are red.
Then we color  $P$ alternately by red and blue when  $t$ is even.
We color $x_1x_2P'$ alternately by red and blue, and color  $x_{t-2}x_{t-1}$ by blue when  $t$ is odd.
It is clear that $T$ is conflict-free 2-edge-colorable.

\noindent{\bf Case 2.} $x_1x_2$ is the conflict-free edge of $E_{T_1}(x_1x_2)$,
but $x_{t-1} x_t$ is not the conflict-free edge of  $E_{T_2}(x_{t-1}x_t)$.

Note that we can give conflict-free edge-colorings to $T_1$ and $T_2$ such that  the colors of $x_1x_2$ and $x_{t-1} x_t$ are red.
Then red is the conflict-free color in $T_1$ and blue is the conflict-free color in $T_2$.
If $t$ is odd, then we color $x_1x_2P'$ alternately by red and blue, and color  $x_{t-2}x_{t-1}$ by red.
If $t$ is even, then we color $P'$ alternately by red and blue such that the color of $x_3x_4$ is blue, and color $x_2x_3$ by blue and color $x_{t-2}x_{t-1}$ by red.
It is clear that $T$ is conflict-free 2-edge-colorable.

\noindent{\bf Case 3.} $x_1x_2$ is not the conflict-free edge of $E_{T_1}(x_1x_2)$ and $x_{t-1} x_t$ is not the conflict-free edge of $E_{T_2}(x_{t-1}x_t)$.

If $t$ is odd, then we give conflict-free edge-colorings to $T_1$ and $T_2$ such that the conflict-free color
 of $x_1x_2$ is blue and the conflict-free color of $x_{t-1}x_t$ is red.
  It follows that the color of $x_1x_2$ is red  and the color of $x_{t-1}x_t$ is blue.
We color $x_2P'x_{t-1}$ alternately by red and blue such that the color of $x_2x_3$ is red.
 It is clear that $T$ is conflict-free 2-edge-colorable.

If $t$ is even, then we give conflict-free edge-colorings to $T_1$ and $T_2$ such that the conflict-free colors
 of $x_1x_2$ and $x_{t-1}x_t$ are red.
 It follows that the colors of $x_1x_2$ and $x_{t-1}x_t$ are blue, respectively.
 We color $x_2P'x_{t-1}$ alternately by red and blue such that the color of $x_2x_3$ is blue.
 It is clear that $T$ is conflict-free 2-edge-colorable.
\end{proof}

\begin{figure}[htbp]
    \centering
    \includegraphics[width=240pt]{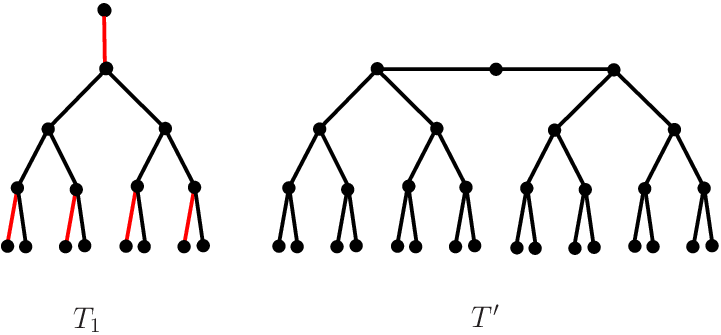}\\
    \caption{The unique conflict-free edge-coloring of $T_1$ and the tree $T'$.} \label{counter1}
\end{figure}
%\vspace{-2cm}

\begin{remark}
If $T_{= 2}$  contains a component of order less than five, then Theorem \ref{thm-main-2} is not true.
For instance, the tree $T_1$ in Figure \ref{counter1} has a unique conflict-free 2-edge-coloring.
Let $T'$ be a tree such that $T_{\geq 3}$ has two components and each component is isomorphic to $T_1$, and $T_{= 2}$ is a $P_3$.
It is clear that $T'$ does not have any conflict-free  coloring with two colors.
Hence, $\chi'_{cf}(T')=3$.
Similarly, the tree $T_2$ in Figure \ref{counter2} has a unique conflict-free 2-edge-coloring.
Let $T''$ be a tree such that $T_{\geq 3}$ has two components and each component is isomorphic to $T_2$, and $T_{= 2}$ is a $P_4$.
It is clear that $T''$ does not have any conflict-free edge-coloring with two colors.
Hence, $\chi'_{cf}(T'')=3$.

Although deciding whether  $\chi'_{cf}(G)=2$ is NP-complete even if $G$ is a bipartite graph \cite{P-Li}, we believe that one can determine whether $\chi'_{cf}(T)=2$ for a tree $T$ in polynomial time.
\end{remark}

\begin{figure}[htbp]
    \centering
    \includegraphics[width=260pt]{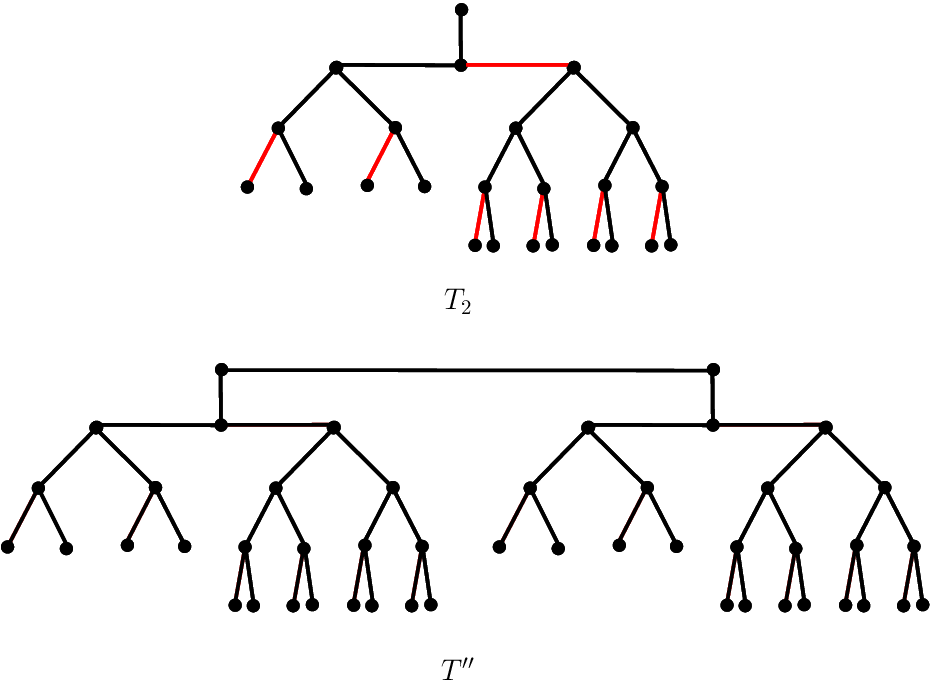}\\
    \caption{The unique conflict-free edge-coloring of $T_2$ and the tree $T''$.} \label{counter2}
\end{figure}

\section{Acknowledgements}
\noindent
Ethan Li is supported by the Fundamental Research Funds for the Central Universities (GK202207023). Ping Li is supported by the National Science Foundation of China No. 12201375.

\end{document}